%% file: main.tex
\documentclass[11pt,letterpaper]{article}

\usepackage[utf8]{inputenc}
\usepackage{amsmath, amsthm, amssymb}
\usepackage[hmargin=3.5cm,vmargin=3.3cm]{geometry}
\usepackage[usenames]{xcolor}
\definecolor{Blue}{RGB}{0, 38, 75}
\definecolor{Red}{RGB}{181,35,17}
\usepackage{hyperref}
\hypersetup{
    colorlinks,
    citecolor=Red,
    filecolor=Blue,
    linkcolor=Blue,
    urlcolor=Blue
}
\usepackage{enumerate}
\usepackage[runin]{abstract}
\usepackage{sectsty}
\usepackage{mathtools}

\sectionfont{\large}
\subsectionfont{\normalsize}

\newtheorem{lemma}{Lemma}
\newtheorem{theorem}{Theorem}
\newtheorem{conjecture}{Conjecture}
\newtheorem{corollary}{Corollary}
\theoremstyle{remark}
\newtheorem*{remark}{Remark}

\DeclareMathOperator{\down}{down}
\DeclareMathOperator{\dist}{dist}

\DeclareMathOperator*{\E}{\mathbb{E}}
\DeclareMathOperator{\EE}{EE}
\DeclareMathOperator{\DISJ}{DISJ}
\DeclareMathOperator{\Match}{Match}
\DeclareMathOperator{\supp}{supp}

\DeclareMathOperator*{\Ent}{H}
\DeclareMathOperator*{\BEnt}{H_2}

\DeclareMathOperator{\D}{\mathbf{D}}
\DeclareMathOperator{\BD}{\mathbf{D}_2}

\makeatletter
\newcommand*{\defeq}{\mathrel{\rlap{%
                     \raisebox{0.3ex}{$\m@th\cdot$}}%
                     \raisebox{-0.3ex}{$\m@th\cdot$}}%
                     =}
\makeatother
\newcommand{\dmid}{\,\|\,}
\newcommand{\emid}{\,|\,}

\newcommand{\propref}[1] {%
\hyperref[#1]{Property~(P\ref*{#1})}}

\begin{document}

\title{\Large\bf On the communication complexity of sparse set disjointness and exists-equal problems}
\author{
\normalsize Mert Sa\u{g}lam\\%
\small University of Washington\\%
\small saglam@google.com \and
\normalsize G\'abor Tardos\\
\small Simon Fraser University and\\
\small Alfr\'ed R\'enyi Institute of Mathematics\\
\small tardos@renyi.hu }
\date{}

\maketitle
\setlength{\abstitleskip}{-\absparindent}
\thispagestyle{empty}
\vspace{-7.5mm}

\input{abstract.tex}

\input{intro.tex}

\input{upperbound.tex}

\input{elementary.tex}

\input {isoperimetry.tex}

\input{lowerbound.tex}

\input{discussion.tex}

\section*{Acknowledgments}
We would like to thank Hossein Jowhari for 
many stimulating discussions during the early stages of this work.

\bibliography{main}{}
\bibliographystyle{plain}
\end{document}

%% file: abstract.tex
\begin{abstract}
In this paper we study the two player randomized communication complexity 
of  the sparse set disjointness and the exists-equal problems and give matching 
lower and upper bounds (up to constant factors) for any number of rounds
for both of these problems. In the sparse set disjointness problem, each player 
receives a $k$-subset of $[m]$ and the goal is to determine whether
the sets intersect. For this problem, we give a protocol that communicates 
a total of $O(k\log^{(r)}k)$ bits over $r$ rounds and errs with very small
probability.
Here we can take $r=\log^{*}k$ to obtain a $O(k)$ total communication 
$\log^{*}k$-round protocol with exponentially small error probability, 
improving on the $O(k)$-bits $O(\log k)$-round constant error
probability protocol of H\aa stad and Wigderson from 1997.

In the exist-equal problem, the players receive vectors $x,y\in [t]^n$ and the
goal is to 
determine whether there exists a coordinate $i$ such that $x_i=y_i$. Namely,
the exists-equal problem is the OR of $n$ equality problems. Observe that 
exists-equal is an instance of sparse set disjointness with $k=n$, hence the protocol 
above applies here as well, giving an $O(n\log^{(r)}n)$ upper bound.
Our main technical contribution in this paper is a matching 
lower bound: we show that when $t=\Omega(n)$, any $r$-round randomized protocol for 
the exists-equal problem with error probability at most $1/3$
should have a message of size $\Omega(n\log^{(r)}n)$. Our lower bound holds 
even for super-constant $r\le \log^*n$,
showing that any $O(n)$ bits exists-equal protocol should 
have $\log^*n - O(1)$ rounds. Note that the protocol we give errs only with
less than polynomially 
small probability and provides guarantees on the total communication for the harder set 
disjointness problem, whereas our lower bound 
holds even for constant error probability protocols and for the easier
exists-equal problem with guarantees on 
the max-communication. Hence our upper and lower bounds match in a strong sense.

Our lower bound on the constant round protocols for exists-equal show that
solving the OR of $n$ instances of the equality problems requires strictly
more than $n$ times the cost of a single instance. To our knowledge this is
the first example of such a {\em super-linear} increase in complexity.
\end{abstract}

%% file: intro.tex
\newpage
\section{Introduction}
In a two player communication problem the players, named Alice and Bob, receive
separate inputs, $x$ and $y$, and they communicate in order to compute the
value $f(x,y)$ of a function $f$. In an $r$-round protocol, the players can
take at most $r$ turns alternately sending each other a message and the last player
to receive a message declares the output of the protocol. A protocol can be
{\em deterministic} or {\em randomized}, in the latter case the players can
base their actions on a common random source and we measure the {\em error
probability}: the maximum over inputs $(x,y)$, of the
probability that the output of the protocol differs from $f(x,y)$.

\subsection{Sparse set disjointness}
Set disjointness is perhaps the most studied problem in 
communication complexity. In the most standard version Alice and Bob 
receive a subset of $[m]:=\{1,\ldots,m\}$ each, with the goal of deciding whether their 
sets intersect or not. The primary question is whether the players can improve on 
the trivial deterministic protocol, where the first player sends the entire input to the other player, 
thereby communicating $m$ bits. The first lower bound on the randomized 
complexity of  this problem
was given in \cite{BabaiFS86} by Babai et al., who showed that any 
$\epsilon$-error protocol
for disjointness must communicate $\Omega(\sqrt{m})$ bits. The tight bound
of $\Omega(m)$-bits was first given by 
Kalyanasundaram and Schnitger \cite{KalyanasundaramS92} and was later
simplified by Razborov \cite{Razborov92} and Bar-Yossef et al.\
\cite{Bar-YossefJKS04}. 

In the sparse set disjointness problem $\DISJ_k^m$, the sets given to the
players are guaranteed to have 
at most $k$ elements. The deterministic communication complexity of this
problem is well understood.
The trivial protocol, where Alice sends her entire input to Bob solves the 
problem in one round using $O(k\log(2n/k))$ bits.
On the other hand, an $\Omega(k\log(2n/k))$ bit total 
communication lower bound can be shown even 
for protocols with an arbitrary number of rounds, say using the rank method; 
see \cite{KushilevitzN97}, page 175.

The randomized complexity of the problem is far more subtle. 
%
%
The results cited above immediately imply a 
$\Omega(k)$ lower bound for this version of the problem.
The folklore $1$-round protocol solves the problem using $O(k\log k)$ bits,
wherein Alice sends $O(\log k)$-bit hashes for each element of her set. 
H\aa stad and Widgerson \cite{HastadW07} gave a protocol that matches the 
$\Omega(k)$ lower bound mentioned above.
Their $O(k)$-bit randomized protocol runs in $O(\log k)$-rounds and
errs  with a small constant probability. In \autoref{sec:upperbound}, we 
improve this protocol to run in $\log^*k$ rounds, still with 
$O(k)$ total communication, but with exponentially small error in $k$. We also
present a $r$-round protocol for any $r<\log^*k$ with total communication
$O(k\log^{(r)}k)$ and error probability well below $1/k$; see
\autoref{thm:ub}. (Here $\log^{(r)}$
denotes the iterated logarithm function, see \autoref{notation}.) As the
exists-equal problem with parameters $t$ and $n$ (see below) is a special case
of $\DISJ_n^{tn}$, our lower bounds for the exists-equal problem (see below)
show that complexity of this algorithm is optimal for any number $r\le\log^*k$
of rounds, even if we allow much the larger error probability of $1/3$.
Buhrman et al.~\cite{BuhrmanGMW12} and Woodruff \cite{Woodruff08} (as presented in \cite{Patrascu09}) show an $\Omega(k\log k)$ lower bound for $1$-round 
complexity of $\DISJ^m_k$ by a reduction from the indexing problem
(a similar reduction was also given in \cite{DasguptaKS12}). 
We note that these lower bounds do not apply to the exists-equal problem,
as the input distribution they use generates instances inherently specific to 
the disjointness problem; 
furthermore this distribution admits a $O(\log k)$ protocol in two rounds.   

\subsection{The exists-equal problem}
In the equality problem Alice and Bob receive elements $x$ and $y$ of a
universe $[t]$ and they have to decide whether $x=y$.
We define the two player communication game exists-equal with parameters 
$t$ and $n$ as follows.
Each player is given an $n$-dimensional vector from $[t]^n$, namely $x$ and $y$.
The value of the game is one if there exists a coordinate
$i\in[n]$ such that $x_i = y_i$, zero otherwise. Clearly, this problem is the
OR of $n$ independent instances of the equality problem.

The direct sum problem in communication complexity is the study of whether
$n$ instances of a problem can be solved using less than 
$n$ times the  communication required for a single instance of the problem.
This question has been studied extensively for specific communication problems
as well as some class of problems \cite{ChakrabartiSWY01, JainRS03, 
JainRS05, Ben-AroyaRW08, Gavinsky08, JainKN08, HarshaJMR10,  BarakBCR10}.
The so called direct sum approach is a very powerful tool to show lower
bounds for communication games. In this approach, one expresses the problem
at hand, say as the OR of $n$ instances of a simpler function
and the lower bound
is obtained by combining a lower bound for the simpler problem 
with a direct sum argument.
For instance, the two-player and multi-player disjointness
bounds of \cite{Bar-YossefJKS04}, the lopsided set disjointness bounds 
\cite{Patrascu11},  and the lower bounds for several communication 
problems that arise from streaming algorithms 
\cite{JayramW09, MagniezMN10} are a few examples of results that 
follow this approach.

Exists-equal with parameters $t$ and $n$ is a special
case of $\DISJ_n^{tn}$, so our protocols in \autoref{sec:upperbound} solve
exists-equal. We show that when $t=\Omega(n)$ these protocols are optimal,
namely every $r$-round randomized protocol ($r\le\log^*n$) with at most $1/3$
error error probability needs to send at least 
one message of size $\Omega(n\log^{(r)}n)$ bits. See \autoref{thm:main}. Our
result shows
that computing the OR of $n$ instances of the equality problem requires
{\em strictly more} than $n$ times the communication required to solve a
single instance
of the equality problem when the number of rounds is smaller than $\log^*
n-O(1)$.
Recall that the equality problem admits an $\epsilon$-error 
$\log(1/\epsilon)$-bit one-round protocol in the common random source model.

For $r=1$, our result implies that to compute the OR of $n$ 
instances of the equality problem with {\em constant probability}, no protocol
can do better than solving each instance of the equality problem with 
{\em high probability} so that the union bound can be applied when taking the
OR of the computed results. The single round case of our lower bound
also generalizes the $\Omega(n\log n)$ lower bound of Molinaro et al.\
\cite{MolinaroWY13} for the
one round communication problem, where the players have to find all the
answers of $n$ equality problems, outputting an $n$ bit string.

\subsection{Lower bound techniques}

We obtain our general lower bound via a round elimination argument. 
In such an argument one assumes the existence of a protocol $P$ that 
solves a communication problem, say $f$, in $r$ rounds. 
By suitably modifying the internals of $P$, one obtains another 
protocol $P'$ with $r-1$ rounds, which typically solves smaller 
instances of $f$ or has larger error than $P$. 
Iterating this process, one obtains a protocol with zero rounds. 
If the protocol we obtain solves non-trivial instances of $f$ with good 
probability, we conclude that we have arrived at a contradiction,
therefore the protocol we started with, $P$, cannot exist. 
Although round elimination arguments have been used 
for a long time, our round elimination lemma is the first to 
prove a {\em super-linear} communication lower bound in the 
number of primitive problems involved, obtaining which requires new 
and interesting ideas.

The general round elimination presented in \autoref{sec:lowerbound} is very 
involved, but the lower bound on
the one-round protocols can also be obtained in a more elementary way. As the
one round case exhibits the most dramatic super-linear increase in the
communication cost and also generalizes the lower bound in \cite{MolinaroWY13}, we
include this combinatorial argument separately in \autoref{sec:elementary},
see \autoref{thm:singleround}.

At the heart of the general round elimination lemma is a new isoperimetric 
inequality on the discrete cube $[t]^n$ endowed with the Hamming distance. We
present this result, \autoref{thm:isoperimetry}, in \autoref{sec:isoperimetry}.
To the best of our knowledge, the first isoperimetric inequality on this metric 
space was proven by Lindsey in \cite{Lindsey64}, where the subsets of $[t]^n$ of a certain size 
with the so called minimum induced-edge number were characterized. 
This result was rediscovered in \cite{KleitmanKR71} and \cite{Clements71} as well.
See \cite{AzizogluO03} for a generalization of this inequality to universes which are $n$-dimensional boxes with arbitrary side lengths.
In \cite{BollobasL91}, Bollobás et al.\ study isoperimetric inequalities on $[t]^n$ 
endowed with the $\ell_1$ distance.
For the purposes of our proof we 
need to find sets $S$ that minimize a substantially more complicated measure. This measure 
also captures how spread out $S$ is and can be described roughly as
the average over points $x\in[t]^n$ of the logarithm of the number of 
points in the intersection of $S$ and a Hamming ball around $x$.

\subsection{Related work}
In \cite{MiltersenNSW98}, a round elimination lemma was given, 
which applies to a class of problems with certain 
self-reducibility properties.
The lemma is then is used to get lower bounds for various problems including 
the greater-than and the predecessor problems.
This result was later tightened in \cite{Sen03} to get better bounds for the 
aforementioned problems.
Different round elimination arguments were also used in
\cite{KarchmerW90, HalstenbergR88, NisanW93,Miltersen94,
DurisGS87,BeameF01} for various 
communication complexity lower bounds and most 
recently in \cite{BrodyC09} and \cite{BrodyCRVW10} for obtaining 
lower bounds for the gapped 
Hamming distance problem.

Independent of and in parallel of the present form of this paper Brody et al.\
\cite{BrodyCK12} have also established an $\Omega(n\log^{(r)}n)$ lower bound
for the $r$-round communication complexity of the exists-equal problem with
parameter $n$. Their result applies for protocols with a polynomially small
error probability like $1/n$. This stronger assumption on the protocol allows
for simpler proof techniques, namely the 
information complexity based direct sum technique developed in 
several papers including \cite{ChakrabartiSWY01},
but it is not enough to create an 
example where solving the OR of $n$ communication problems requires more than
$n$ times the communication of solving a single instance. Indeed, even in the
shared random source model one needs $\log n$ bits of communication
(independent of the number of rounds) to achieve $1/n$ error in a single
equality problem.

\subsection{Notation}\label{notation}
For a positive integer $t$, we write $[t]$ for the set of positive integers not
exceeding $t$. For two $n$-dimensional vectors 
$x$, $y$, let $\Match(x,y)$ be the number of coordinates where $x$ and $y$
agree. Notice that $n-\Match(x,y)$ is the Hamming distance between $x$ and
$y$.
For a vector $x\in[t]^n$ we write $x_i$ for its $i$\/th coordinate.
We denote the distribution of a random variable $X$ by $\dist(X)$ and 
the support set of it by $\supp(X)$. 
We write $\Pr_{x\sim\nu}[\cdot]$ and
$\E_{x\sim\nu}[\cdot]$ for
the probability and expectation, respectively, when $x$ is distributed
according to a distribution $\nu$. We write $\mu$ for the uniform distribution
on $[t]^n$.  For instance, for a set $S\subseteq [t]^n$, we have $\mu(S) = |S| / t^n$.

For $x,y\in [t]^n$ we denote 
the value of the exists-equal game by $\EE_n^t(x,y)$. Recall that it is zero if
and only if  $x$ and $y$ differ in each coordinate. Whenever we drop $t$ from
the notation we assume $t=4n$. Often we will also drop $n$
and simply denote the game value by $\EE(x,y)$ if $n$ is clear from the 
context.

All logarithms in this paper are to the base 2. 
Analogously, throughout this paper we take $\exp(x)=2^x$.
We will also use the iterated versions of these functions:
\begin{align*}
\log^{(0)}x&\defeq x, & \exp^{(0)}x&\defeq x,\\
\log^{(r)}x&\defeq \log(\log^{(r-1)}x), & \exp^{(r)}x&\defeq \exp(\exp^{(r-1)}x)
\quad\text{for $r\geq 1$}.
\end{align*}
Moreover we define $\log^* x$ to be the smallest integer $r$ for which
$\log^{(r)} x<2$. 

Throughout the paper we ignore divisibility problems, e.g., in
\autoref{lem:determine-s} in \autoref{sec:elementary} we assume that
$t^n/2^{c+1}$ is an integer. Dealing with rounding issues would complicate the
presentation but does not add to the complexity of the proofs.

\subsection{Information theory}
Here we briefly review some definitions and facts from information theory 
that we use in this paper.
For a random variable $X$, we denote its binary Shannon entropy 
by $\Ent(X)$. We will also use conditional entropies
$\Ent(X\emid Y)=\Ent(X,Y)-\Ent(Y)$. Let $\mu$ and $\nu$ be two probability 
distributions, supported on the same set 
$S$. We denote the binary Kullback-Leibler 
divergence between $\mu$ and $\nu$ by
$\D(\mu\dmid \nu)$. A random variable with Bernoulli distribution with
parameter $p$ takes the value $1$ with probability $p$ and the value $0$ with
probability $1-p$. The entropy of this variable is denoted by $\BEnt(p)$.
For two reals $p,q\in (0,1)$, we denote by $\BD(p\dmid q)$ the divergence
between the Bernoulli distributions with parameters $p$ and $q$.

If $X\in[t]^n$ and $L\subseteq[n]$, then the projection of $X$ to the
coordinates in $L$ is denoted by $X_L$. Namely, $X_L$ is obtained from
$X=(X_1,\ldots,X_n)$ by keeping only the coordinates $X_i$ with $i\in L$.
The following lemma of 
Chung et al.~\cite{ChungGFS86} relates the entropy of a variable to the
entropy of its projections.
\begin{lemma}
\label{lem:ent-subset}{\rm (Chung et al.~\cite{ChungGFS86})}
Let $\supp(X)\subseteq[t]^n$. We have
$\frac{l}{n}\Ent(X) \leq \E_L[\Ent(X_L)]$, where the expectation is taken for
a uniform random $l$-subset $L$ of $[n]$.
\end{lemma}

\subsection{Structure of the paper}

We start in \autoref{sec:upperbound} with our protocols for the sparse set
disjointness. Note that the exists-equal problem is a special case of sparse
set disjointness, so our protocols work also for the exists-equal problem. In
the rest of the paper we establish matching lower bounds showing that the
complexity of our protocols are within a constant factor to optimal for both
the exists-equal and the sparse set disjointness problems, and for any number of rounds. 
In \autoref{sec:elementary} we give an elementary proof for the case of
single round protocols. In \autoref{sec:isoperimetry} we develop our
isoperimetric inequality and in \autoref{sec:lowerbound} we use it in our round
elimination proof to get the lower bound for multiple round protocols. Finally in
\autoref{sec:discussion} we point toward possible extensions of our
results.

%% file: upperbound.tex
\section{The upper bound}
\label{sec:upperbound}
Recall that in the communication problem $\DISJ_k^m$,
each of the two players
is given a subset of $[m]$ of size at most $k$
and they communicate in order to determine
whether their sets are disjoint or not.
In 1997, H\aa stad and Wigderson \cite{ParnafesIRWA97,HastadW07} gave a
probabilistic protocol that solves this problem with $O(k)$ bits of
communication and has constant one-sided error probability.
The protocol takes $O(\log k)$ rounds.
Let us briefly review this protocol as this is the starting point of
our protocol.

Let $S,T\subseteq[m]$ be the inputs of Alice and Bob. Observe that if they
find a set $Z$ satisfying $S\subseteq Z\subseteq [m]$, then Bob can replace
his input $T$ with $T'=T\cap Z$ as $T'\cap S=T\cap S$. The main observation is
that if $S$ and $T$ are disjoint, 
then a random set $Z\supseteq S$ will intersect $T$ in a uniform random
subset, so one can expect $|T'|\approx|T|/2$. In the H\aa stad-Wigderson
protocol the players
alternate in finding a random set that contains the current input of one of
them, effectively halving the other player's input. If in this process
the input of one of the players becomes empty, they know the original inputs
were
disjoint. If, however, the sizes of their inputs do not show the expected
exponential decrease in time, then they declare that their inputs
intersect. This introduces a small one sided error. Note that one of the two
outcomes happens in $O(\log k)$ rounds. An important observation is that
Alice can describe a random set $Z\supseteq S$ to
Bob using an expected $O(|S|)$ bits by making use of the joint random source.
This makes the total communication $O(k)$.

In our protocol proving the next theorem, we do almost the same, but we choose
the random sets $Z\supseteq S$ not uniformly, but from a biased distribution
favoring ever smaller sets. This makes the size of the input sets of the
players decrease much more rapidly, but describing the random set $Z$
to the other player becomes more costly. By carefully balancing the parameters we
optimize for the total communication given any number of rounds. When the
number of rounds reaches $\log^*k-O(1)$ the communication reaches its minimum
of $O(k)$ and the error becomes exponentially small.

\begin{theorem}\label{thm:ub} For any $r\leq \log^*k$, there is an
$r$-round probabilistic protocol for $\DISJ^m_k$ with $O(k\log^{(r)}k)$ bits
total communication. There is no error for intersecting input sets, and the
probability of error for disjoint sets can be made $O(1/\exp^{(r)}(c\log^{(r)}
k)+ \exp(-\sqrt k))\ll 1/k$ for any constant $c > 1$.

For $r=\log^*k-O(1)$ rounds this means an $O(k)$-bit protocol with error
probability $O(\exp(-\sqrt k))$.
\end{theorem}

\begin{proof} We start with the description of the protocol. Let $S_0$ and
  $S_1$ be the input sets of Alice and Bob,
respectively. For $1\le i\le r$, $i$ even Alice sends a message describing
a set $Z_i\supset S_i$ based on her ``current input'' $S_i$ and Bob updates
his ``current input'' $S_{i-1}$ to $S_{i+1}\defeq S_{i-1}\cap Z_i$. In odd numbered
rounds the same happens with the role of Alice and Bob reversed. We depart
from the H\aa stad-Wigderson protocol in the way we choose the sets $Z_i$:
Using the shared random source the players generate $l_i$ random subsets of
$[m]$ containing each element of $[m]$ independently and with probability
$p_i$. We will set these parameters later. The set $Z_i$ is chosen to be
the first such set containing $S_i$. Alice or Bob (depending on the parity of
$i$) sends the index of this set or ends the protocol by sending a
special error signal if none of the generated sets contain $S_i$. The protocol
ends with declaring the inputs disjoint if the error signal is never sent and
we have $S_{r+1}=\emptyset$. In all other cases the protocol ends with
declaring ``not disjoint''.

This finishes the description of the protocol except for the setting of the
parameters. Note that the error of the protocol is one-sided:
$S_0\cap S_1=S_i\cap S_{i+1}$ for $i\le r$, so intersecting inputs
cannot yield $S_{r+1}=\emptyset$.

We set the parameters (including $k_i$ used in the analysis) as follows:
\begin{align*}
u&=(c+1)\log^{(r)}k,\\
p_i&=\frac1{\exp^{(i)}u}&\hbox{for }1\le i\le r,\\
l_1&=k\exp(ku),\\
l_i&=k2^{k/2^{i-4}}&\hbox{for }2\le i\le r,\\
k_0&=k_1=k,\\
k_i&=\frac k{2^{i-4}\exp^{(i-1)}u}&\hbox{for }2\le i\le r,\\
k_{r+1}&=0.
\end{align*}

The message sent
in round $i>1$ has length $\lceil\log(l_i+1)\rceil<k/2^{i-4}+\log k+1$, thus
the total communication in all rounds but the first is $O(k)$.
The length of the first message is $\lceil\log(l_1+1)\rceil\le ku+\log
k+1$. The total communication is $O(ku)=O(ck\log^{(r)}k)$ as claimed
(recall that $c$ is a constant).

Let us assume the input pair is disjoint. To estimate the error
probability we call round $i$ {\em bad} if an error message is sent or a set
$S_{i+1}$ is created with $|S_{i+1}|>k_{i+1}$. If no bad round exists we have
$S_{r+1}=\emptyset$ and the protocol makes no error. In what follows we bound
the probability that round $i$ is bad assuming the previous rounds are not bad
and therefore having $|S_j|\le k_j$ for $0\le j\le i$.

The probability that a random set constructed in round $i$ contains $S_i$
is $p_i^{-|S_i|}\ge p_i^{-k_i}$. The probability that none of the $l_i$
sets contains $S_i$ and thus an error message is sent is therefore at most
$(1-p_i^{k_i})^{l_i}<e^{-k}$.

If no error occurs in the first bad round $i$, then $|S_{i+1}|>k_{i+1}$. Note that
in this case $S_{i+1}=S_{i-1}\cap Z_i$ contains each element of $S_{i-1}$
independently and with probability $p_i$. This is because the choice of $Z_i$
was based on it containing $S_i$, so it was independent of its intersection
with $S_{i-1}$ (recall that $S_i\cap S_{i-1}=S_1\cap S_0=\emptyset$). For
$i<r$ we use the Chernoff bound. The expected size of $S_{i+1}$ is
$|S_{i-1}|p_i\le k_{i-1}p_i\le k_{i+1}/2$, thus the probability of
$|S_{i+1}|>k_{i+1}$ is at most $2^{-k_{i+1}/4}$. Finally for the last round
$i=r$ we use the simpler estimate $p_rk_{r-1}\le k/\exp^{(r)}u$ for
$|S_{r+1}|>k_{r+1}=0$.

Summing over all these estimates we obtain the following error bound for our
protocol:
$$\Pr[\hbox{error}]\le re^{-k}+\frac k{\exp^{(r)}u}+\sum_{i=2}^r2^{-k_i/4}.$$
In case $k_r\ge4\sqrt n$ this error estimate proves the theorem. In case
$k_r<4\sqrt k$ we need to make a minor adjustments in the setting of our
parameters. We take $j$ to be the smallest value with $k_j<4\sqrt k$, modify
the parameters for round $j$ and stop the protocol after this round declaring
``disjoint'' if $S_{j+1}=\emptyset$ and ``intersecting'' otherwise. The new
parameters for round $j$ are $k'_j=4\sqrt k$, $p'_j=2^{-2\sqrt k}$,
$l'_j=k2^{8k}$. This new setting of the parameters makes the message in the
last round linear in $k$, while both the probability that round $j-1$ is bad
because it makes $|S_j|>k'_j$, or the probability that round $j$ is bad for any
reason (error message or $S_{j+1}\ne\emptyset$) is $O(2^{-\sqrt k})$. This
finishes the analysis of our protocol.
\end{proof}

%% file: elementary.tex
\section{Lower bound for single round protocols}
\label{sec:elementary}
In this section we give an combinatorial proof that any single round randomized
protocol for the exists-equal problem with parameters $n$ and $t=4n$ has
complexity $\Omega(n\log n)$ if its error probability is at most $1/3$. As
pointed out in the Introduction, to our knowledge this is the
fist established case when solving the OR of $n$ instances of a communication
problem requires strictly more than $n$ times the complexity needed to solve a
single such instance.

We start with with a simple and standard reduction from the randomized
protocol to the deterministic one and further to a large set of inputs that
makes the first (and in this case only) message fixed. These steps are also
used in the general round elimination argument therefore we state them in
general form.

Let $\epsilon>0$ be a small constant and let $P$ be an $1/3$-error 
randomized protocol for
the exists-equal problem with parameters $n$ and $t=4n$. We repeat the
protocol $P$ in parallel taking the majority output, so that the number of
rounds does not change, the length of the messages is multiplied by a constant
and the error probability decreases below $\epsilon$. Now we fix the coins of
of this $\epsilon$-error protocol in a way to make the resulting deterministic
protocol err on at most $\epsilon$ fraction of the possible inputs.
Denote the deterministic protocol we obtain by $Q$.

\begin{lemma}
\label{lem:determine-s}
Let $Q$ be a deterministic protocol for the $\EE_n$ problem that makes at most
$\epsilon$ error on the uniform distribution.
Assume Alice sends the first message of length
$c$. There exists an $S\subset [t]^n$ of size $\mu(S)=2^{-c-1}$ such 
that the first message  of Alice is fixed when $x\in S$ and we have
$\Pr_{y\sim \mu}[Q(x,y)\neq\EE(x,y)]\leq 2\epsilon$ for all $x\in S$.
\end{lemma}

\begin{proof}
Note that the quantity $e(x)=\Pr_{y\sim \mu}[Q(x,y)\neq\EE(x,y)]$, 
averaged over all $x$, is the error probability of $Q$ on the uniform input,
hence is at most $\epsilon$. Therefore for at least half of $x$,
we have $e(x)\leq 2\epsilon$. The first message of Alice partitions this half
into
at most $2^c$ subsets. We pick $S$ to consist of $t^n/2^{c+1}$ vectors of the
same part: at least one part must have this many elements.
\end{proof}

We fix a set $S$ as guaranteed by the lemma. We assume we started with a
single round protocol, so $Q(x,y)=Q(x',y)$ whenever $x,x'\in S$. Indeed, Alice
sends the same message by the choice of $S$ and then the output is determined
by Bob, who has the same input in the two cases.

We call a pair $(x,y)$ {\em bad} if $x\in S$, $y\in[t]^n$ and $Q$ errs on this
input, i.e., $Q(x,y)\ne\EE(x,y)$. Let $b$ be the number of bad pairs. By
\autoref{lem:determine-s} each $x\in|S|$ is involved in at most $2\epsilon
t^n$ bad pairs, so we have
$$b\le2\epsilon|S|t^n.$$
We call a triple
$(x,x',y)$ {\em bad} if $x,x'\in S$, $y\in[t]^n$, $\EE(x,y)=1$ and
$\EE(x',y)=0$. The proof is based on double counting the number $z$ of bad
triples.
Note that for a bad triple $(x,x',y)$ we have $Q(x,y)=Q(x',y)$ but
$\EE(x,y)\ne\EE(x',y)$, so $Q$ must err on either $(x,y)$ or $(x',y)$ making
one of these pairs bad. Any pair (bad or not) is involved in at most $|S|$ bad
triples, so we have
$$z\le b|S|\le2\epsilon|S|^2t^n.$$

Let us fix arbitrary $x,x'\in S$ with $\Match(x,x')\le n/2$. We estimate the
number of $y\in[t]^n$ that makes $(x,x',y)$ a bad triple. Such a $y$ must have
$\Match(x,y)>\Match(x',y)=0$. To simplify the calculation we only count the
vectors $y$ with $\Match(x,y)=1$. The match between $y$ and $x$ can occur at
any position $i$ with $x_i\ne x'_i$. After fixing the coordinate $y_i=x_i$ we
can pick the remaining coordinates $y_j$ of $y$ freely as long as we avoid
$x_j$ and $x'_j$. Thus we have
$$|\{y\emid(x,x'y)\hbox{ is
  bad}\}|\ge(n-\Match(x,y))(t-2)^{n-1}\ge(n/2)(t-2)^{n-1}>t^n/14,$$ 
where in the last inequality we used $t=4n$. Let $s$ be the size of the
Hamming ball $B_{n/2}(x)=\{y\in[t]^n\emid\Match(x,y)>n/2\}$. By the Chernoff
bound we have $s<t^n/n^{n/2}$ (using $t=4n$ again).
For a fixed $x$ we have at least
$|S|-s$ choices for $x'\in S$ with $\Match(x,x')\le n/2$ when the above bound
for triples apply. Thus we have
$$z\ge|S|(|S|-s)t^n/14.$$
Combining this with the lower bound on the number of bad triples we get
$$28\epsilon|S|\ge|S|-s.$$

Therefore we conclude that we either have large error $\epsilon>1/56$ or
else we have $|S|\le2s<2t^n/n^{n/2}$. As we have $|S|=t^n/2^{c+1}$ the latter
possibility implies
$$c\ge n\log n/2-2.$$
Summarizing we have the following.

\begin{theorem}
\label{thm:singleround}
A single round probabilistic protocol for $\EE_n$ with error probability $1/3$
has complexity $\Omega(n\log n)$.

A single round deterministic protocol for $\EE_n$ that
errs on at most $1/56$ fraction of the inputs has complexity at least $n\log
n/2-2$.
\end{theorem}

%% file: isoperimetry.tex
\section{An isoperimetric inequality on the discrete grid}
\label{sec:isoperimetry}

The isoperimetric problem on the 
Boolean cube $\{0,1\}^n$ proved
extremely useful in theoretical computer science. 
The problem is to determine the set $S\subseteq \{0,1\}^n$ of a fixed
cardinality with the smallest ``perimeter'', or more generally, to establish
connection between the size of a set and the size of its boundary. 
Here the boundary can be defined in several ways. 
Considering the Boolean cube as a graph
where vertices of Hamming distance 1  are connected, 
the {\em edge boundary} of a set $S$ is defined as 
the set of edges connecting $S$ and its complement,
while the {\em vertex boundary} consists of the vertices 
outside $S$ having a neighbor in $S$.

Harper \cite{Harper64} showed that the vertex boundary
of a Hamming ball is smallest among all sets of equal size, and the same holds 
for the edge boundary of a subcube. 
These results can be generalized to other cardinalities \cite{Hart76}; 
see the survey by Bezrukov \cite{Bezrukov94}.

Consider the metric space over the set $[t]^n$ endowed 
with the Hamming distance.
Let $f$ be a concave function on the nonnegative integers 
and $1\le M<n$ be an integer. 
We consider the following
value as a generalized perimeter of a set $S\subseteq[t]^n$:
\begin{align*}
\E_{x\sim\mu}[f\left(\left|B_M(x)\cap S\right|\right)],
\end{align*}
where $B_M(x)=\{y\in[t]^n\mid\Match(x,y)\ge M\}$ is the radius $n-M$ Hamming
ball around $x$. Note that when $M=n-1$ and  $f$ is the counting function
given as $f(0)=0$ and $f(l)=1$ for $l>0$ (which is concave),
the above quantity is exactly the
normalized size of the vertex boundary of $S$. For other concave functions $f$
and parameters $M$ this quantity can still be considered a measure of how
``spread out'' the set $S$ is. We conjecture that $n$-dimensional boxes
minimize this measure in every case.

\begin{conjecture}
\label{conj:product}
Let $1\le k\le t$ and $1\le M<n$ be integers. 
Let $S$ be an arbitrary subset of $[t]^n$ of size $k^n$ and $P=[k]^n$.
We have
\begin{align*}
\E_{x\sim\mu}[f\left(\left|B_M(x)\cap P\right|\right)]\leq
\E_{x\sim\mu}[f\left(\left|B_M(x)\cap S\right|\right)].
\end{align*}
\end{conjecture}
Even though a proof of \autoref{conj:product} remained elusive, 
in \autoref{thm:isoperimetry}, 
we prove an approximate version of this result, where, for technical reasons,
we have to restrict our attention to a small fraction of the coordinates. Having
this weaker result allows us to prove our communication complexity lower bound
in the next section but proving the conjecture here would simplify this proof.

We start the technical part of this section by introducing the notation we
will use.
For $x,y\in[t]^n$ and $i\in[n]$ we write $x\sim_iy$ if $x_j=y_j$ for
$j\in[n]\setminus\{i\}$. Observe that $\sim_i$ is an equivalence relation.
A set $K\subseteq [t]^n$ is called an {\em $i$-ideal} if $x\sim_i y$,
$x_i<y_i$ and $y\in K$ implies $x\in K$. We call a set $K\subseteq[t]^n$ an {\em ideal} if it is an
$i$-ideal for all $i\in[n]$.
%

For $i\in[n]$ and $x\in[t]^n$ we define
$\down_{i}(x)=(x_1,\ldots,x_{i-1},x_i-1,x_{i+1},\ldots,x_n)$. We have
$\down_i(x)\in[t]^n$ whenever $x_i>1$.
Let $K\subseteq [t]^n$ be a set, $i\in[n]$ and $2\le a\in[t]$. For
$x\in K$, we define $\down_{i,a}(x,K)=\down_i(x)$ if $x_i=a$ and
$\down_i(x)\notin K$ and we set $\down_{i,a}(x,K)=x$ otherwise. 
We further define
$\down_{i,a}(K)=\{\down_{i,a}(x,K)\mid x\in K\}$.
For $K\subseteq[t]^n$ and $i\in[n]$ we define
\begin{align*}
\down_i(K)=\big\{y\in[t]^n \mid y_i\le|\{z\in K\mid y\sim_iz\}|\big\}.
\end{align*}
Finally for $K\subseteq[t]^n$ we define
\begin{align*}
\down(K)=\down_1(\down_2(\ldots\down_n(K)\ldots)).
\end{align*}
The following lemma states few simple observations about these down
operations.
\begin{lemma}
\label{lem:down}
Let $K\subseteq[t]^n$ be a set and let $i,j\in[n]$ be integers. 
The following hold.
\begin{enumerate}[(i)]
\item $\down_i(K)$ can be obtained from $K$ by applying several
  operations $\down_{i,a}$.
\item $|\down_{i,a}(K)|=|K|$ for each $2\le a\le t$, $|\down_i(K)|=|K|$ and
$|\down(K)|=|K|$.
\item $\down_i(K)$ is an $i$-ideal and
if $K$ is a $j$-ideal, then $\down_i(K)$ is also a $j$-ideal.
\item $\down(K)$ is an ideal. For any $x\in\down(K)$ we have $P\defeq[x_1]\times[x_2]\times\cdots\times[x_n]\subseteq\down(K)$
  and there exists a set $T\subseteq K$ with
  $P= \down(T)$.
\end{enumerate}
\end{lemma}
\begin{proof}
For statement (i) notice that as long as $K$ is not 
an $i$-ideal one of the operations $\down_{i,a}$ 
will not fix $K$ and hence will decrease $\sum_{x\in K}x_i$. 
Thus a finite sequence of these operations will transform $K$ into
an $i$-ideal. 
It is easy to see that the operations $\down_{i,a}$ preserve the
number of elements in each equivalence class of $\sim_i$, 
thus the $i$-ideal we arrive at must indeed be $\down_i(K)$.

Statement (ii) follows directly from the definitions of each of
these $\down$ operations.
 
The first claim of statement (iii), namely that $\down_i(K)$ is an $i$-ideal,
is trivial from the definition. 
Now assume $j\ne i$ and $K$ is a $j$-ideal, $y\in\down_i(K)$ and $y_j>1$.
To see that $\down_i(K)$ is a $j$-ideal it is enough to prove that
$\down_j(y)\in\down_i(K)$. Since $y\in\down_i(K)$,
there are $y_i$ distinct vectors $z\in K$ that satisfy $z\sim_i y$. 
Considering the vectors $\down_j(z)\sim_i\down_j(y)$ and
using that these distinct vectors are in the $j$-ideal $K$ proves that
$\down_j(y)$ is indeed contained in $\down_i(K)$.

By statement (iii), 
$\down(K)$ is an $i$-ideal for each $i\in [n]$.
Therefore $\down(K)$ is an ideal and the first part of statement (iv), that is, 
$P\subseteq K'$ follows.
We prove the existence of suitable $T$ by induction on the
dimension $n$. The base case $n=0$ (or even $n=1$) is trivial. For the
inductive step consider $K'=\down_2(\down_3(\ldots\down_n(K)\ldots))$.
As $x\in\down(K)=\down_1(K')$, we have distinct vectors
$x^{(k)}\in K'$ for $k=1,\ldots, x_1$, satisfying $x^{(k)}\sim_1x$.
Notice
that the construction of $K'$ from $K$ is performed 
independently on each of the $(n-1)$-dimensional ``hyperplanes''
$S^l=\{y\in[t]^n\mid y_1=l\}$ as none of the operations
$\down_2,\ldots,\down_n$ change the first coordinate of the vectors.
We apply
the inductive hypothesis to obtain the sets
$T^{(k)}\subseteq S^{x^{(k)}_1}\cap K$
such that
$\down_2(\ldots\down_n(T^{(k)})\ldots)=\{x^{(k)}_1\}
\times[x_2]\times\cdots\times[x_n]$. Using
again that these sets are in distinct hyperplanes and the operations
$\down_2,\ldots,\down_n$ act separately on the hyperplanes $S^l$, we get for
$T:=\cup_{k=1}^{x_1}T^{(k)}$ that
$$\down_2(\dots\down_n(T)\dots)=\{x^{(k)}_1\mid
k\in[x_1]\}\times[x_2]\times\cdots\times[x_n].$$
Applying $\down_1$ on both sides finishes the proof of this last part of the
lemma.
\end{proof}

For sets $x\in[t]^n$, $I\subseteq[n]$, 
and integer $M\in[n]$ we define
$B_{I,M}(x)=\{y\in[t]^n\mid\Match(x_I,y_I)\ge M\}$.
The projection of $B_{I,M}$ to the coordinates in $I$ is the Hamming ball of
radius $|I|-M$ around the projection of $x$.
\begin{lemma}
\label{lem:list}
Let $I\subseteq[n]$, $M\in[n]$ and let $f$ be a concave function on the
nonnegative integers. For arbitrary $K\subseteq [t]^n$ we have
$$\E_{x\sim\mu}[f(|B_{I,M}(x)\cap\down(K)|)]\le\E_{x\sim\mu}[f(|B_{I,M}(x)\cap
K|)].$$
\end{lemma}

\begin{proof}
By \autoref{lem:down}(i), the set $\down(K)$ can be obtained from 
$K$ by a series of operations 
$\down_{i,a}$ with various $i\in[n]$ and $2\le a\le t$.
Therefore, it is enough to prove that the expectation in the
lemma does not increase in any one step. 
Let us fix $i\in[n]$ and $2\le a\le t$. 
We write $N_x=B_{I,M}(x)\cap K$ and
$N'_x=B_{I,M}(x)\cap\down_{i,a}(K)$ for $x\in[t]^n$. We need to prove that
\begin{align*}
\E_{x\sim\mu}[f(|N_x|)]\ge\E_{x\sim\mu}[f(|N'_x|)].
\end{align*} 
Note that $|N_x|=|N'_x|$ whenever $i\notin I$ or $x_i\notin\{a,a-1\}$.
Thus, we can assume $i\in I$ and concentrate on $x\in[t]^n$ with
$x_i\in\{a,a-1\}$. It is enough to prove $f(|N_x|)+f(|N_y|)\ge
f(|N'_x|)+f(|N'_y|)$ for any pair of vectors $x,y\in[t]^n$, 
satisfying $x_i=a$, and $y=\down_i(x)$.

Let us fix such a pair $x,y$ and set $C=\{z\in
K\setminus\down_{i,a}(K)\emid\Match(x_I,z_I)=M\}$. 
Observe that $N_x = N'_x \cup C$ and $N'_x\cap C=\emptyset$.
Similarly, observe that $N'_y = N_y \cup \down_{i,a}(C)$ and 
$N_y \cap \down_{i,a}(C)=\emptyset$.
Thus we have $|N'_x|=|N_x|-|C|$ and 
$|N'_y|=|N_y|+|\down_{i,a}(C)|=|N_y|+|C|$.

The inequality $f(|N_x|)+f(|N_y|)\ge f(|N'_x|)+f(|N'_y|)$ follows now from
the concavity of $f$, the inequalities $|N'_x|\le|N_y|\le|N'_y|$ and the equality $|N_x|+|N_y|=|N'_x|+|N'_y|$. Here the first inequality follows from
$\down_{i,a}(N'_x)\subseteq\down_{i,a}(N_y)$, the second inequality and the
equality comes from the observations of the previous paragraph.
\end{proof}

\begin{lemma}
\label{lem:find-prod}
Let $K\subseteq [t]^n$ be arbitrary. 
There exists a vector $x\in K$ having at least $n/5$ coordinates 
that are greater than $k\defeq\frac{t}{2}\mu(K)^{5/(4n)}$.
\end{lemma}
\begin{proof}
The number of vectors that have at most $n/5$ coordinates greater
than $k$ can be upper bounded as 
\begin{align*}
{n\choose n/5} t^{n/5} k^{4n/5} 
= t^n {n\choose n/5} (k/t)^{4n/5} 
= |K|\frac{{n \choose n/5}}{2^{4n /5}},
\end{align*}
where in the last step we have substituted
$\frac{k}{t}=\frac{1}{2}\mu(K)^{5/(4n)}$ and $\mu(K) = |K| / t^n$.
Estimating ${n\choose n/5}\le 2^{n\BEnt(1/5)}$, 
we obtain that the above quantity is less than $|K|$.
Therefore, there must exists an $x\in K$ that has at least $n/5$ coordinates
greater than $k$.
\end{proof}

\begin{theorem}
\label{thm:isoperimetry}
Let $S$ be an arbitrary subset of  $[t]^n$.  Let $k=\frac{t}{2}\mu(S)^{5/(4n)}$ 
and $M = nk/(20t)$. There exists a subset $T\subset S$ of size $k^{n/5}$ 
and $I\subset [n]$ of size $n/5$ such that, defining 
$N_x=\{x'\in T\mid\Match(x_I,x'_I)\ge M\}$,
we have
\begin{enumerate}[(i)]
\item $\Pr_{x\sim\mu}[N_x=\emptyset] \le 5^{-M}$ and
\item $\E_{x\sim \mu}[\log|N_x|]\geq (n/5-M)\log k - n\log k /5^M$,
where we take $\log 0 = -1$ to make the above expectation exist.
\end{enumerate}
\end{theorem}
\begin{proof}
By \autoref{lem:down}(ii), we have $|\down(S)|=|S|$. 
By \autoref{lem:find-prod}, there exists an $x\in\down(S)$ 
having at least $n/5$ coordinates that are greater than $k$.
Let $I\subset[n]$ be a set of $n/5$ coordinates such that 
$x_i\geq k$ for a fixed $x\in\down(S)$. By \autoref{lem:down}(iv),
$\down(S)$ is an ideal and thus it contains the set $P=\prod_iP_i$,
where $P_i=[k]$ for $i\in I$ and $P_i=\{1\}$ for $i\notin I$. 
Also by \autoref{lem:down}(iv), there exists a $T\subseteq S$ such that
$P = \down(T)$. We fix such a set $T$. Clearly, $|T|=k^{n/5}$.

For a vector $x\in[t]^n$, let $h(x)$ be the number of coordinates $i\in I$ 
such that $x_i\in [k]$. Note that $\E_{x\sim \mu}[h(x)] = 4M$ and $h(x)$ has a
binomial distribution. By the Chernoff bound we have $\Pr_{x\sim \mu}[h(x)<M]
< 5^{-M}$. For $x$ with $h(x)\ge M$ we have $|B_{I,M}(x)\cap P|\ge
k^{n/5-M}$, but for $h(x)<M$ we have $B_{I,M}(x)\cap P=\emptyset$. With the
unusual convention $\log0=-1$ we have
\begin{align*}
\E_{x\sim \mu} [\log|B_{I,M}(x)\cap P|]&\ge\Pr[h(x)\ge M](n/5-M)\log
k-\Pr[h(x)<M]\\
&>(n/5-M)\log k-n\log k/5^M
\end{align*}

We have $\down(T)=P$ and our unusual $\log$ is concave on the nonnegative
integers, so \autoref{lem:list} applies and proves statement (ii):
\begin{align*}
\E_{x\sim \mu}[\log |N_x|] &\ge\E_{x\sim \mu} [\log|B_{I,M}(x)\cap P|]\\
&\ge(n/5-M)\log k - n\log k /5^M.
\end{align*}

To show statement (i), we apply \autoref{lem:list} with the concave function
$f$ defined as $f(0)=-1$ and  $f(l)=0$ for all $l>0$. We obtain that
\begin{align*}
\Pr_{x\sim\mu}[N_x=\emptyset]
&=-\E_{x\sim\mu}[f(|N_x|)]\\
&\le-\E_{x\sim\mu}[f(|B_{I,M}(x)\cap P|)]\\
&=\Pr_{x\sim\mu}[B_{I,M}(x)\cap P=\emptyset]\\
&<5^{-M}.
\end{align*}
This completes the proof.
\end{proof}

%% file: lowerbound.tex
\section{Lower bound for multiple round protocols}
\label{sec:lowerbound}



In this section we prove our main lower bound result:
\begin{theorem}
\label{thm:main}
For any $r\leq\log^*n$,
an $r$-round probabilistic protocol for $\EE_n$ with error probability at most
$1/3$ sends at least one message of size 
$\Omega(n\log^{(r)}n)$.
\end{theorem}

Note that the $r=1$ round case of this theorem was proved as
\autoref{thm:singleround} in \autoref{sec:elementary}. 
The other extreme, which immediately follows from \autoref{thm:main}, is the following.

\begin{corollary}
Any probabilistic protocol for $\EE_n$ with maximum message size $O(n)$ and
error $1/3$ has at least $\log^* n - O(1)$ rounds.
\end{corollary}

\autoref{thm:main} is a direct consequence of the corresponding statement on 
deterministic protocols with small distributional error on uniform distribution; 
see \autoref{thm:main2} at the end of this section. 
Indeed, we can decrease the error of a randomized protocol below any 
constant $\epsilon>0$ for the price of increasing the message length by a 
constant factor, then we can fix the coins of this low error protocol in a 
way that makes the resulting
deterministic protocol $Q$ err in at most $\epsilon$ fraction of the possible
inputs. Applying \autoref{thm:main2} to the protocol $Q$ proves
\autoref{thm:main}.

In the rest of this section we use round-elimination to prove
\autoref{thm:main2}, that is, we will use $Q$ to solve smaller instances of the 
exists-equal problem in a way that the first message
is always the same, and hence can be  eliminated.

Suppose Alice sends the first message of $c$ bits in $Q$. By \autoref{lem:determine-s}, 
there exists a $S\subset [t]^n$ of size $\mu(S)=2^{-c-1}$ such that 
 the first message  of Alice is fixed when $x\in S$ and we have
$\Pr_{y\sim \mu}[Q(x,y)\neq\EE(x,y)]\leq 2\epsilon$ for all $x\in S$.
Fix such a set $S$ and let $k\defeq t/2^{\frac{5(c+1)}{4n} + 1}$ and $M \defeq nk/(20t)$.
By \autoref{thm:isoperimetry}, there exists a $T\subset S$ of size 
$k^{n/5}$ and $I\subset[n]$ of size $n/5$ such that defining
\begin{align*}
N_x=\{y\in T\mid\Match(x_I,y_I)\ge M\}
\end{align*}
we have
$\Pr_{x\sim\mu}[N_x=\emptyset] \le 5^{-M}$ and
$\E_{x\sim \mu}[\log|N_x|]\geq (n/5-M)\log k - n\log k /5^M$.
Let us fix such sets $T$ and $I$. 
Note also that \autoref{thm:isoperimetry} guarantees that $T$ is 
a strict subset of $S$.  Designate an arbitrary element of $S\setminus T$ 
as $x'_e$.

\subsection{Embedding the smaller problem}
\label{sec:embed}
The players embed a smaller
instance $u,v\in[t']^{n'}$ of the exists-equal problem in $\EE_n$
concentrating on the 
coordinates $I$ determined above.
We set $n'\defeq M/10$ and
$t'\defeq4n'$. Optimally, the same embedding should guarantee low error
probability for all pairs of inputs, but for technical reasons we need to know
the number of coordinate agreements $\Match(u,v)$ for the input pairs $(u,v)$
in the smaller problem having $\EE_{n'}(u,v)=1$. Let $R\ge1$ be this number, so we
are interested in inputs $u,v\in[t']^{n'}$ with $\Match(u,v)=0$ or $R$. We
need this extra parameter so that we can eliminate a non-constant number of
rounds and still keep the error bound a constant. For results on
constant round protocols one can concentrate on the $R=1$ case.

In order to solve the exist-equal problem
with parameters $t'$ and $n'$ Alice and Bob use the joint random source to
turn their input $u,v\in[t']^{n'}$ into longer random vectors $X',Y\in[t]^n$,
respectively, and apply the protocol $Q$ above to solve this exists-equal
problem for these larger inputs. Here we informally list the main requirements
on the process generating $X'$ and $Y$. We require these properties for the
random vectors $X',Y\in[t]^n$ generated from a fixed pair $u,v\in[t']^{n'}$
satisfying $\Match(u,v)=0$ or $R$.

\begin{enumerate}[(P1)]\item $\EE(X',Y)=\EE(u,v)$ with large probability,
\label{prop:2}
\item $\supp(X')=T\cup \{x'_e\}$ and
\label{prop:3}
\item for most $x'\sim X'$, we have $\dist(Y\emid X'=x')$ is 
close to uniform distribution on $[t]^n$.
\label{prop:4}
\end{enumerate}

Combining these properties with the fact that 
$\Pr_{y\sim \mu}[Q(x,y)\neq\EE(x,y)]\leq 2\epsilon$ for each $x\in S$, 
we will argue that for the considered pairs of inputs $Q(X',Y)$ equals
$\EE(u,v)$ with large probability, thus the combined protocol solves the small
exists-equal instance with small error, at least for input pairs with
$\Match(u,v)=0$ or $R$. Furthermore, by \propref{prop:3} the first
message of Alice will be fixed and hence does not need to be sent, making the
combined protocol one round shorter.

The random variables $X'$ and $Y$ are constructed as follows. Let $m\defeq
2n/(MR)$ be an integer. Each player repeats his or her input ($u$ and $v$,
respectively) $m$ times, obtaining a vector of size $n/(5R)$. 
Then using the shared randomness, 
the players pick $n/(5R)$ uniform random maps
$m_i:[t']\to[t]$ independently  and apply $m_i$ to $i$\/th
coordinate. Furthermore, the players pick a uniform random 1-1 mapping
$\pi:[n/(5R)]\to I$ and use it to embed the
coordinates of the vectors they constructed among the coordinates of the
vectors $X$ and $Y$ of length $n$. The remaining $n-n/(5R)$ coordinates of $X$
is picked uniformly at random by Alice and similarly, the remaining $n-n/(5R)$
coordinates of $Y$ is picked uniformly at random by Bob. Note that the
marginal distribution of both $X$ and $Y$ are uniform on $[t]^n$. If
$\Match(u,v)=0$ the vectors $X$ and $Y$ are independent, while if
$\Match(u,v)=R$, then $Y$ can be obtained by selecting a random subset of $I$
of cardinality $mR$, copying the corresponding coordinates of $X$ and filling
the rest of $Y$ uniformly at random.

This completes the description 
of the random process for Bob. However Alice generates one more 
random variable $X'$ as follows. 
Recall that $N_x=\{z\in T\mid\Match(z_I,x_I)\ge M\}$.
The random variable $X'$ is obtained by drawing $x\sim X$ first and 
then choosing a uniform random element of $N_x$.
In the (unlikely) case that $N_x=\emptyset$, Alice chooses $X'=x'_e$.

Note that $X'$ either equals $x'_e$ or takes values from $T$, hence
\propref{prop:3} holds.
In the next lemma we quantify and prove \propref{prop:2} as well.

\begin{lemma}
\label{lem:error}
Assume $n\ge3$, $M\ge2$ and $u,v\in[t']^{n'}$. We have
\begin{enumerate}[(i)]
\item 
if $\Match(u,v)=0$ then $\Pr[\EE(X',Y)=0] > 0.77$;
\item 
if $\Match(u,v)=R$, then
$\Pr[\EE(X', Y)=1] \ge 0.80$.
\end{enumerate}
\end{lemma}

\begin{proof}
For the first claim, note that when 
$\Match(u,v) = 0$, the random variables $X$ and $Y$ are independent and 
uniformly distributed. We construct $X'$ based on $X$, so its value
is also independent of $Y$. Hence $\Pr[\EE(X',Y)=0]=(1-1/t)^n$.
This quantity goes to $e^{-1/4}$ since $t=4n$ and is larger than $0.77$
when $n\geq 3$. This establishes the first claim.

For the second claim let $J = \{i\in I\mid X_i=Y_i\}$ and $K=\{i\in I\mid X'_i = X_i\}$.
By construction, $|J|=\Match(X_{I},Y_{I})\ge mR$ 
and $|K|=\Match(X'_{I}, X_I) \geq M$ unless $N_X=\emptyset$.
By our construction, each $J\subset I$ 
of the same size is equally likely by symmetry, even when we condition on a
fix value of $X$ and $X'$. Thus we have $\E[|J\cap
K|\emid N_X\ne\emptyset]\ge mRM/|I|=10$ and $\Pr[J\cap K=\emptyset\emid
N_X\ne\emptyset]<e^{-10}$. Note that $X$ is distributed
uniformly over $[t]^n$, therefore by Theorem~\ref{thm:isoperimetry}(i) the
probability that $N_X=\emptyset$ is at most $5^{-M}$. 
Note that $\Match(X',Y)\ge|J\cap K|$ and thus $\Pr[\EE(X',Y)=0]\le\Pr[J\cap
K=\emptyset]\le\Pr[J\cap K=\emptyset\emid
N_X\ne\emptyset] +\Pr[N_X=\emptyset]\le e^{-10}+5^{-M}$. This completes the proof.
\end{proof}

We measure ``closeness to uniformity'' in \propref{prop:4} by simply
calculating the entropy. This entropy argument is postponed to
the next subsection; here we show how such a bound to the entropy implies that
the error introduced by $Q$ is small.

\begin{lemma}
\label{lem:kl-err}
Let $x'\in S$ be fixed and let $\gamma$ be a probability in the range
$2\epsilon\le\gamma<1$. If $\Ent(Y\emid X'=x')\geq n\log t - 
\BD(\gamma \dmid 2\epsilon)$ then $\Pr_{y\sim Y|X'=x'}[Q(x',y)
\neq \EE(x',y)]\leq \gamma$.
\end{lemma}
\begin{proof}
For a distribution $\nu$ over $[t]^n$, let 
$e(\nu) = \Pr_{y\sim \nu}[Q(x', y)\neq\EE(x', y)]$. 
We prove the contrapositive of the
statement of the lemma, that is assuming $\Pr_{y\sim Y|X'=x'}[Q(x',y) 
\neq \EE(x',y)]>\gamma$ we prove $\Ent(Y\emid X'=x')< n\log t - 
\BD(\gamma \dmid 2\epsilon)$:
\begin{align*}
n\log t -\Ent(Y\emid X' = x') &= \D(\dist(Y\emid X'=x')\dmid \mu)\\
&\geq \BD(e(\dist(Y\emid X'=x'))\dmid e(\mu))\\
&\geq\BD(\gamma\dmid 2\epsilon),
\end{align*}
where the first inequality follows from the chain rule for the 
Kullback-Leibler divergence.
\end{proof}



\subsection{Establishing \propref{prop:4}}

We quantify \propref{prop:4} using the conditional entropy
$\Ent(Y\emid X')$. If $\Match(u,v)=R$ our process generetas $X$ and $Y$ with
the expected number $\E[\Match(X_I,Y_I)]$ of matches only slightly more than the
minimum $mR$. We lose most of these matches with $Y$ when we replace $X$ by
$X'$ and only an expected constant number remains. A constant number of forced
matches with $X'$ within $I$ restricts the number of possible vectors $Y$ but
it only decreases the entropy by $O(1)$. The calculations in
this subsection make this intuitive argument precise.

\begin{lemma}
\label{lem:y-uniform}
Let $X',Y$ be as constructed above. The following hold.
\begin{enumerate}[(i)]
\item If $\Match(u,v)=0$ we have
$\Ent(Y\emid X')=n\log t.$
\item If $M>100\log n$ and $\Match(u,v)=R$ we have
$\Ent(Y\emid X') = n\log t - O(1).$
\end{enumerate}
\end{lemma}
\begin{proof}
Part (i) holds as $Y$ is uniformly distributed and independent of $X'$
whenever $\EE(u,v)=0$.

For part (ii) recall that if $\Match(u,v)=R$ one can construct $X$ and $Y$ by
uniformly selecting a size $mR$ set $L\subseteq I$ and selecting $X$ and $Y$
uniformly among all pairs satisfying $X_L=Y_L$. Recall that $L$ is the set of
coordinates the $mR$ matches between
$u^m$ and $v^m$ were mapped. These are the ``intentional matches'' between
$X_I$ and $Y_I$. Note that there may be also ``unintended matches'' between
$X_I$ and $Y_I$, but not too many: their expected number is
$(n/5-mR)/t<1/20$. As given any fixed
$L$, the marginal distribution of both $X$ and $Y$ are still uniform, so in
particular $X$ is independent of $L$ and so is $X'$ constructed from
$X$. Therefore we have
$$\Ent(Y\emid X') =\Ent(Y\emid X', L) + \Ent(L) - \Ent(L\emid Y, X').$$
We treat the terms separately. First we split the first term:
$$\Ent(Y\emid X',L)=\Ent(Y_L\emid X',L)+\Ent(Y_{[n]\setminus L}\emid
X',L,Y_L)$$
and use that $Y_{[n]\setminus L}$ is uniformly distributed for any fixed
$L$, $X'$ and $Y_L$, making
$$\Ent(Y_{[n]\setminus L}\emid X',L,Y_L)=(n-mR)\log t.$$
We have $X_L=Y_L$, thus
\begin{align*}
\Ent(Y_L\emid X',L)&=\Ent(X_L\emid X',L)\\
&\ge \frac{mR}{n/5}\Ent(X_I\emid X')\\
&\ge mR\log t-10\log k-\frac{MR}{5^{M-1}}\log k,
\end{align*}
where the first inequality follows by \autoref{lem:ent-subset} as $L$ is a
uniform and independent of $X$ and $X'$ and the second inequality follows from
\autoref{lem:x-uniform} that we will prove shortly and the formula defining
$m$.

The next term, $\Ent(L)$ is easy to compute as $L$ is a uniform subset of $I$
of size $mR$:
$$\Ent(L)=\log{n/5\choose mR}$$

It remains to bound the term $\Ent(L\emid Y, X')$. 
Let $Z=\{i\mid i\in I \text{ and } X'_i=Y_i\}$.
Note that $Z$ can be derived from $X',Y$ (as $I$ is fixed) 
hence $\Ent(L\emid Y,X')\leq \Ent(L\emid Z)$.
Further, let $C=|Z\setminus L|$. We obtain
\begin{align*}
\Ent(L\emid Y,X')&\le\Ent(L\emid Z)\leq \Ent(L\emid Z, C) + \Ent(C)\\
&< \E_{Z,C}\left[\log{n/5-|Z|+C \choose mR - |Z|+C}\right]
+ \E_{Z,C}\left[\log {|Z| \choose C}\right]+2
\end{align*}
where we used $\Ent(C)<2$.
Note that for any fixed $x'\in T$ and $x\in \supp(X\emid X'=x')$, we have 
$$\E[|Z|-C\emid X=x, X'=x']=\Match(x_I,x_I') mR /(n/5) \geq 10$$
as $\Match(x_I, x_I')\geq M$ by definition. 

Hence we have
$$\log{n/5\choose mR}-\log{n/5-|Z|+|C|\choose mR-|Z|+|C|}\ge10\log\frac
n{5m}-O(1),$$
$$\E_{Z,C}\left[\log {|Z| \choose C}\right] \le \E[|Z|] < 20.$$
Summing the estimates above for the various parts of $\Ent(Y\emid X')$ the
statement of the lemma follows.
\end{proof}

It remains to prove the following simple lemma that ``reverses'' the
conditional entropy bound in Theorem~\ref{thm:isoperimetry}(ii):

\begin{lemma}
\label{lem:x-uniform} For any $u,v\in[t']^{n'}$ we have
$\Ent(X_I\emid X')\geq \frac{n}{5}\log t - M\log k  - n\log k / 5^M$.
\end{lemma}
\begin{proof}
Using the fact that $\Ent(A,B) = \Ent(A\emid B) +\Ent(B) = \Ent(B \emid A) + \Ent(A)$ we get
\begin{align*}
\Ent(X_I\emid X') &= \Ent(X' \emid X_I) + \Ent(X_I) - \Ent(X')\\
&\ge \frac{n}{5}\log t + \Ent(X'\emid X_I) - \frac{n}{5}\log k,
\end{align*}
where in the last step we used $\Ent(X')\le \log|\supp(X')| 
= \log |T| = \frac{n}{5}\log k$ and $\Ent(X_I)=(n/5)\log t$ 
as $X$ is uniformly distributed.

Observe that $\Ent(X'\emid X_I)=\Ent(X'\emid X) = \E_{x\sim \mu}[\log|N_x|]$,
where $\log 0$ is now taken to be $0$.
From \autoref{thm:isoperimetry}(ii) we get 
$\Ent(X'\emid X)\geq \frac{n}{5}\log k - M\log k - n\log k/5^M$ finishing the
proof of the lemma.
\end{proof}

\subsection{The round elimination lemma}

Let $\nu_n$ be the uniform distribution on $[t]^n\times [t]^n$, where we set $t=4n$.
The following lemma gives the base case of the round elimination argument.
\begin{lemma}\label{lem:terminal} 
Any 0-round deterministic protocol for $\EE_n$ has at least 0.22 
distributional error on $\nu_n$, when $n\geq 1$.
\end{lemma}
\begin{proof}
The output of the protocol is decided by a single player, say Bob. For any
given input $y\in[t]^n$ we have
$3/4 \leq\Pr_{x\sim\mu}[\EE(x,y)=0] < e^{-1/4} < 0.78$. 
Therefore the distributional error is at least $0.22$ for any given $y$
regardless of the output Bob chooses, thus the overall error is also at
least $0.22$.
\end{proof}
Now we give our full round elimination lemma.
\begin{lemma}\label{lem:roundel} 
Let $r>0, c ,n$ be an integers such that $c < (n\log n)/2$. 
There is a constant $0<\epsilon_0<1/200$ such that 
if there is an $r$-round deterministic protocol with $c$-bit messages for 
$\EE_n$ that has $\epsilon_0$ error on $\nu_n$, 
then there is an $(r-1)$-round deterministic protocol with $O(c)$-bit messages for 
$\EE_{n'}$ that has $\epsilon_0$ error on $\nu_{n'}$, 
where $n' = \Omega(n/2^\frac{5c}{4n})$.
\end{lemma}
\begin{proof}
We start with an intuitive description of our reduction. Let us be given the
deterministic protocol $Q$ for $\EE_n$ that errs on an $\epsilon_0$ fraction
of the inputs. To solve an instance $(u,v)$ of the smaller $\EE_{n'}$ problem
the players perform the embedding procedure
described in previous subsection $k_0$ times independently for each parameter
$R\in[R_0]$. Here $k_0$ and $R_0$ are constants we set later. They perform the
protocol $Q$ in parallel for each of the $k_0R_0$ pairs of
inputs they generated. Then they take the majority of the $k_0$ outputs for a
fixed parameter $R$. We show that this result gives the correct value of
$\EE(u,v)$ with large probability provided that $\Match(u,v)=0$ or
$R$. Finally they take the OR of these results for the $R_0$ possible values
of $R$. By the union bound this gives the correct value $\EE(u,v)$ with large
probability provided $\Match(u,v)\le R_0$. Fixing the random choices of the
reduction we obtain a deterministic protocol. The probability of error for the
uniform random input can only grow by the small probability that
$\Match(u,v)>R_0$ and we make sure it remains below $\epsilon_0$.
The rest of the proof makes this argument precise.

For random variables $X'$ and $Y$ constructed in \autoref{sec:embed}, 
\autoref{lem:y-uniform} guarantees that 
$\Ent(Y\emid X')\ge n\log t - \alpha_0$ 
for some constant $\alpha_0$, as long as $M>100\log n$ and $\Match(u,v)=R$. Let $\epsilon_0$
be a constant such that $\BD(1/10\dmid 2\epsilon_0) > 
200(\alpha_0 + 1)$.   
Note that such $\epsilon_0$ can be found as $\BD(1/10\dmid \epsilon)$ 
tends to infinity as $\epsilon$ goes to 0.  
We can bound 
$\Pr_{(x,y)\sim\nu_m}[\Match(x,y) \ge l] \le 1/(4^l l!)$ for all $m\ge1$.
We set $R_0$ such that 
$\Pr_{(x,y)\sim\nu_m}[\Match(x,y) \ge R_0 ] \le \epsilon_0 / 2$  for all $m\ge1$.

Let $Q$ be a deterministic protocol for $\EE_n$ that sends $c < (n\log n)/2$
in each round 
and that has $\epsilon_0$ error on $\nu_n$.
Let $S$ be as constructed in \autoref{lem:determine-s} and let $M$ be
as defined in \autoref{thm:isoperimetry}. 
We have $M=\frac{n}{40}2^{\frac{-5(c+1)}{4n}}$ as $t=4n$ 
and $\mu(S)=2^{-(c+1)}$ by \autoref{lem:determine-s}.
Note that by our choice of $c$, we have $M>100\log n$, hence the 
hypotheses of \autoref{lem:y-uniform} are satisfied.

Let $n' = M/10 = \frac{n}{400}2^{\frac{-5(c+1)}{4n}}$. 
Now we give a randomized protocol $Q'$ for $\EE_{n'}$.
Suppose the players are given an instance of $\EE_{n'}$, 
namely the vectors $(u,v)\in[4n']^{n'}\times[4n']^{n'}$.
Let
$k_0 = 10\log (R_0 + 1/\epsilon_0)$.
For $R\in[R_0]$ and $k\in [k_0]$, the players construct the 
vectors $X'_{R,k}$ and $Y_{R,k}$ as described in \autoref{sec:embed} 
with parameter $R$ and with fresh randomness for each of the $R_0k_0$
procedures.
The players run $R_0 k_0$ instances of protocol $Q$ in parallel, 
on inputs $X'_{R,k}, Y_{R,k}$ for $R\in[R_0]$ and $k\in[k_0]$.
Note that the first message of the first player, Alice, is fixed for all
instances of $Q$ by \propref{prop:3} and \autoref{lem:determine-s}. 
Therefore, the second player, Bob, can start the protocol assuming 
Alice has sent the fixed first message.
After the protocols finish, for each $R\in[R_0]$, the last player who received
a message computes $b_R$ as the majority of $Q(X_{R,k}',Y_{R,k})$
for $k\in [k_0]$. Finally, this player outputs $0$ if $b_R=0$ for all
$R\in[R_0]$ and outputs $1$ otherwise. 

Suppose now that $\EE(u,v) = 0$. By \autoref{lem:error}(i),
we have $\Pr[\EE(X'_{R,k},Y_{R,k}) = 0] \ge 0.77$ for each $R$ and $k$.
Recall that that $Y_{R,k}$ is 
distributed uniformly for each $R$ and $k$ and since $\EE(u,v)=0$, 
it is independent of $X'_{R,k}$.
Therefore, by $X'_{R,k}\in S$ (\propref{prop:3}) and the fact that 
$\Pr_{y\sim \mu}[Q(x,y)\neq\EE(x,y)]\leq 2\epsilon_0$ 
for all $x\in S$ as per \autoref{lem:determine-s},
we obtain $\Pr[Q(X'_{R,k},Y_{R,k}) = 0] \ge 0.77 - 2\epsilon_0 > 0.76$.
By the Chernoff bound we have $\Pr[b_R = 1]  < \epsilon_0/(2R_0)$, and by the
union bound $\Pr[Q'\hbox{ outputs }0]\ge1-\epsilon_0 /2$. 

Let us now consider the case $\Match(u,v) = R$ for some $R\in[R_0]$. Fix any
$k\in[k_o]$ and set $X'=X'_{R,k}$, $Y=Y_{R,k}$.
By \autoref{lem:error}(ii), $\Pr[\EE(X',Y) = 1]\ge 0.80$.
By \autoref{lem:y-uniform}, 
$\Ent(Y\emid X')\geq n\log t -\alpha_0$,
and so we have $\E_{x'\sim X'}[\Ent(Y) - \Ent(Y\emid X'=x')] < 
\alpha_0$. Let $Z=\{x'\mid\Ent(Y) - \Ent(Y\emid X'=x') 
> 10\alpha_0\}$. Note that $Y$ is uniform, and has full entropy, 
therefore $\Ent(Y) - \Ent(Y\emid X'=x') \geq 0$.
Using Markov's inequality we have $\Pr[X'\in Z]<1/10$.
When $X'\in Z$ we cannot effectively bound the probability that 
$\EE(u,v)\neq Q(X', Y)$; namely, we bound this probability by 1.
But if $X'\notin Z$, then 
by \autoref{lem:kl-err} and our choice of $\epsilon_0$, we have 
$\Pr[\EE(X', Y)\neq Q(X', Y)] < 1/10$. Furthermore, by 
\autoref{lem:error}(ii), $\Pr[\EE(u,v) \neq \EE(X', Y)]< 0.20$
 hence with probability at least $0.60$ we have 
$\EE(u,v) = Q(X', Y)$. This happens independently for all the values of
$k\in[k_0]$, so by the Chernoff bound and our choice of $k_0$, we have
$\Pr[Q'\hbox{ outputs }0]\le\Pr[b_R = 0] < \epsilon_0 / 2$.

Finally, $\Pr_{(u,v)\sim \nu_{n'}}[\Match(u,v) \ge R_0] \le \epsilon_0 /2$ by
our choice of $R_0$.
Note that the protocol $Q'$ uses a shared random bit string, say $W$, in the
construction of the vectors $X'_{R,k}$ and $Y_{R,k}$.
Hence, overall, we have
\begin{align*}
\Pr_{W, (u,v)\sim\nu_{n'}}[\EE(u,v)  = Q'(u,v)] \ge 1 - \epsilon_0
\end{align*}
Since we measure the error of the protocol under a distribution, 
we can fix $W$ to a value without 
increasing the error under the aforementioned distribution by the so called easy direction of 
Yao's lemma. Namely, there exists a $w\in \supp(W)$ such that
\begin{align*}
\Pr_{(u,v)\sim\nu_{n'}}[\EE(u,v)  = Q'(u,v)\emid W=w] \ge 1 - \epsilon_0
\end{align*}
Fix such $w$. Observe that $Q'$ is a $(r-1)$-round protocol for $\EE_{n'}$ where 
$n'=\frac{n}{400}2^\frac{-5(c+1)}{4n}=\Omega(n/2^\frac{5c}{4n})$ 
and it sends at most $R_0k_0c = O(c)$ bits in each message.
Furthermore, $Q'$ is deterministic and
has at most $\epsilon_0$ error on $\nu_{n'}$ as desired.
\end{proof}

\begin{theorem}
\label{thm:main2}
There exists a constant $\epsilon_0$ such that for any $r\leq\log^*n$,
an $r$-round deterministic protocol for $\EE_n$ which has
$\epsilon_0$ error on $\nu_n$ sends at least one message of size 
$\Omega(n\log^{(r)}n)$.
\end{theorem}
\begin{proof}
Suppose we have an $r$-round protocol with $c$-bit messages for $\EE_n$ that has $\epsilon_0$ error on $\nu_n$, where $c=\gamma n\log^{(r)}n$ for some $\gamma<4/5-o(1)$. By \autoref{lem:roundel}, this protocol can be converted to an
$r-1$ round protocol with $\alpha c$-bit messages for $\EE_{n'}$ that has 
$\epsilon_0$-error on $\nu_{n'}$, where $n'=\beta n/2^{5c/4n}$ for some $\alpha, \beta >0$. 
 We
only need to verify that $\alpha c \leq \gamma n'\log^{(r-1)}n'$. We have
\begin{align*}
\gamma n'\log^{(r-1)} n' &= \gamma\beta n/2^{5c/4n}\log^{(r-1)}
(\beta n/2^{5c/4n})\\
&= \gamma\beta n/2^{\frac{5\gamma}{4}\log^{(r)}n}\log^{(r-1)}
(\beta n/2^{5c/4n})\\
&\geq \gamma\beta n \left(\log^{(r-1)}n\right)^{1-\frac{5\gamma}{4}-o(1)}\\
&\geq \gamma\alpha n\log^{(r)}n
\end{align*}
for $\gamma< 4/5 - o(1)$ and large enough $n$. Therefore, by iteratively applying
\autoref{lem:roundel} we obtain a $0$-round protocol
for $\EE_{\bar n}$ that makes $\epsilon_0$ error on 
$\nu_{\bar n}$ for some $\bar  n$ satisfying 
$\gamma {\bar n}^2 = \gamma \bar n \log^{(0)} \bar n\geq c \alpha^r$. 
Therefore
 $\bar n \geq 1$ and since $\epsilon_0< 0.22$, the protocol we obtain contradicts \autoref{lem:terminal}, showing that the protocol we started with cannot exists.
\end{proof}
\begin{remark} We note that in the proof of \autoref{thm:main}, 
to show that a protocol with small communication does not exist, 
we start with the given protocol and apply the round elimination lemma 
(i.e., \autoref{lem:roundel}) $r$ times to obtain a $0$-round protocol with small
error probability, which is shown to be impossible by \autoref{lem:terminal}.
Alternatively, one can apply the round elimination $r-1$ times to obtain a 
$1$-round protocol with $o(n\log n)$ communication for $\EE_{n}$, 
which is ruled out by \autoref{thm:singleround}.
\end{remark}

%% file: discussion.tex
\section{Discussion}
\label{sec:discussion}
The $r$-round protocol we gave in \autoref{sec:upperbound} solves the sparse set
disjointness problem in $O(k\log^{(r)}k)$ total communication. As we proved in
\autoref{sec:lowerbound} this is optimal. The same, however,
  cannot be said of the error probability. With the same protocol, but with more
  careful setting of the parameters the exponentially small error $O(2^{-\sqrt
    k})$  of the $\log^*k$-round protocol can be further decreased to
  $2^{-k^{1-o(1)}}$.

For small (say, constant) values of $r$ this protocol
  cannot achieve exponentially small error error without the increase in the
  complexity if the universe size $m$ is unbounded. But if $m$ is polynomial in
  $k$ (or even slightly larger, $m=\exp^{(r)}(O(\log^{(r)}k))$), we can replace the last round of the protocol by one player
deterministically sending his or her entire ``current set'' $S_r$. With careful
setting of the parameters in other rounds, this modified protocol has the same
$O(k\log^{(r)}k)$ complexity but the error is now exponentially small:
$O(2^{-k/\log k})$. Note that in our lower bound on the $r$-round complexity
of the sparse set disjointness we we use the exists-equal problem with
parameters $n=k$ and $t=4k$. This corresponds to the universe size
$m=tn=4k^2$. In this case any protocol
solving the exists-equal problem with $1/3$ error can be strengthened to
exponentially small error using
the same number of rounds and only a constant factor more communication.

Our lower and upper bounds match for the exists-equal problem with parameters
$n$ and $t=\Omega(n)$, since the upper bounds were established without any
regard of the universe size, while the lower bounds worked for $t=4n$.
Extensions of the techniques presented in this paper give matching bounds also
in the case $3\le t<n$, where the $r$-round complexity is
$\Theta(n\log^{(r)}t)$ for $r\le\log^*t$. Note, however, that in this case one
needs to consider significantly more complicated input distributions
and a more refined isoperimetric inequality, that does not permit arbitrary mismatches. 
The $\Omega(n)$ lower bound applies for the exists-equal problem of parameters $n$
and $t\ge3$ regardless of the number of rounds, as the disjointness problem on a
universe of size $n$ is a sub-problem. For $t=2$ the situation
is drastically different, the exists-equal problem with $t=2$ is equivalent to
a single equality problem.

Finally a remark on using the joint random source model of randomized
protocols throughout the paper. By a result of Newman \cite{Newman91} our
protocols of \autoref{sec:upperbound} can be made to work in private coin model 
(or even if one of the players is forced
to behave deterministically) by increasing the first message length by
$O(\log\log(N)+\log(1/\epsilon))$ bits, where 
$N= {m \choose k}$ is the number of possible
inputs. In our case this means adding the term 
$O(\log\log m)+o(k)$ to our bound of $O(k\log^{(r)}k)$, since our protocols
make at least $\exp(-k/\log k)$ error.
This additional cost is
insignificant for reasonably small values of $m$, but it is necessary for large
values as the equality problem, which is an instance of disjointness, requires
 $\Omega(\log \log m)$-bits in the private coin model.

Note also that we achieve a super-linear increase in the communication for 
OR of $n$ instances of equality even in the private coin model for $r=1$. 
For $r\geq 2$, no such increase happens in the private coin model as communication
complexity of $\EE^t_n$ is at most $O(n\log\log t)$ however a single 
equality problem requires $\Omega(\log \log t)$ bits.

%% file: main.bbl
\begin{thebibliography}{10}

\bibitem{AzizogluO03}
M.~Azizoğlu and Ö. Öğecioğlu.
\newblock Extremal sets minimizing dimension-normalized boundary in {H}amming
  graphs.
\newblock {\em SIAM Journal on Discrete Mathematics}, 17(2):219--236, 2003.

\bibitem{BabaiFS86}
L{\'a}szl{\'o} Babai, Peter Frankl, and Janos Simon.
\newblock Complexity classes in communication complexity theory (preliminary
  version).
\newblock In {\em FOCS}, pages 337--347. IEEE Computer Society, 1986.

\bibitem{Bar-YossefJKS04}
Ziv Bar-Yossef, T.~S. Jayram, Ravi Kumar, and D.~Sivakumar.
\newblock An information statistics approach to data stream and communication
  complexity.
\newblock {\em J. Comput. Syst. Sci.}, 68(4):702--732, 2004.

\bibitem{BarakBCR10}
Boaz Barak, Mark Braverman, Xi~Chen, and Anup Rao.
\newblock How to compress interactive communication.
\newblock In Schulman \cite{DBLP:conf/stoc/2010}, pages 67--76.

\bibitem{BeameF01}
Paul Beame and Faith~E. Fich.
\newblock Optimal bounds for the predecessor problem and related problems.
\newblock {\em Journal of Computer and System Sciences}, 65:2002, 2001.

\bibitem{Ben-AroyaRW08}
Avraham Ben-Aroya, Oded Regev, and Ronald de~Wolf.
\newblock A hypercontractive inequality for matrix-valued functions with
  applications to quantum computing and ldcs.
\newblock In {\em FOCS}, pages 477--486. IEEE Computer Society, 2008.

\bibitem{Bezrukov94}
Sergei Bezrukov.
\newblock Isoperimetric problems in discrete spaces.
\newblock In {\em Bolyai Soc. Math. Stud}, pages 59--91, 1994.

\bibitem{BollobasL91}
Béla Bollobás and Imre Leader.
\newblock Edge-isoperimetric inequalities in the grid.
\newblock {\em Combinatorica}, 11(4):299--314, 1991.

\bibitem{BrodyC09}
Joshua Brody and Amit Chakrabarti.
\newblock A multi-round communication lower bound for gap {H}amming and some
  consequences.
\newblock In {\em IEEE Conference on Computational Complexity}, pages 358--368.
  IEEE Computer Society, 2009.

\bibitem{BrodyCK12}
Joshua Brody, Amit Chakrabarti, and Ranganath Kondapally.
\newblock Certifying equality with limited interaction.
\newblock {\em Electronic Colloquium on Computational Complexity (ECCC)},
  19:153, 2012.

\bibitem{BrodyCRVW10}
Joshua Brody, Amit Chakrabarti, Oded Regev, Thomas Vidick, and Ronald de~Wolf.
\newblock Better {G}ap-{H}amming lower bounds via better round elimination.
\newblock In Maria~J. Serna, Ronen Shaltiel, Klaus Jansen, and Jos{\'e} D.~P.
  Rolim, editors, {\em APPROX-RANDOM}, volume 6302 of {\em Lecture Notes in
  Computer Science}, pages 476--489. Springer, 2010.

\bibitem{BuhrmanGMW12}
Harry Buhrman, David Garcia-Soriano, Arie Matsliah, and Ronald de~Wolf.
\newblock The non-adaptive query complexity of testing k-parities, 2012.

\bibitem{ChakrabartiSWY01}
Amit Chakrabarti, Yaoyun Shi, Anthony Wirth, and Andrew Chi-Chih Yao.
\newblock Informational complexity and the direct sum problem for simultaneous
  message complexity.
\newblock In {\em FOCS}, pages 270--278. IEEE Computer Society, 2001.

\bibitem{ChungGFS86}
F.~R. Chung, R.~L. Graham, P.~Frankl, and J.~B. Shearer.
\newblock Some intersection theorems for ordered sets and graphs.
\newblock {\em J. Comb. Theory Ser. A}, 43(1):23--37, September 1986.

\bibitem{Clements71}
G.~F. Clements.
\newblock Sets of lattice points which contain a maximal number of edges.
\newblock {\em Proc. Amer. Math. Soc.}, 27:13--15, 1971.

\bibitem{DasguptaKS12}
Anirban Dasgupta, Ravi Kumar, and D.~Sivakumar.
\newblock Sparse and lopsided set disjointness via information theory.
\newblock In {\em APPROX-RANDOM}, pages 517--528, 2012.

\bibitem{DurisGS87}
Pavol Duris, Zvi Galil, and Georg Schnitger.
\newblock Lower bounds on communication complexity.
\newblock {\em Inf. Comput.}, 73(1):1--22, 1987.

\bibitem{Gavinsky08}
Dmitry Gavinsky.
\newblock On the role of shared entanglement.
\newblock {\em Quantum Information {\&} Computation}, 8(1):82--95, 2008.

\bibitem{HalstenbergR88}
Bernd Halstenberg and R{\"u}diger Reischuk.
\newblock On different modes of communication (extended abstract).
\newblock In Janos Simon, editor, {\em STOC}, pages 162--172. ACM, 1988.

\bibitem{Harper64}
L.H. Harper.
\newblock Optimal assignment of numbers to vertices.
\newblock {\em J. Soc. Ind. Appl. Math.}, 12:131--135, 1964.

\bibitem{HarshaJMR10}
Prahladh Harsha, Rahul Jain, David~A. McAllester, and Jaikumar Radhakrishnan.
\newblock The communication complexity of correlation.
\newblock {\em IEEE Transactions on Information Theory}, 56(1):438--449, 2010.

\bibitem{Hart76}
Sergiu Hart.
\newblock A note on the edges of the n-cube.
\newblock {\em Discrete Mathematics}, 14(2):157 -- 163, 1976.

\bibitem{HastadW07}
Johan H{\aa}stad and Avi Wigderson.
\newblock The randomized communication complexity of set disjointness.
\newblock {\em Theory of Computing}, 3(1):211--219, 2007.

\bibitem{JainKN08}
Rahul Jain, Hartmut Klauck, and Ashwin Nayak.
\newblock Direct product theorems for classical communication complexity via
  subdistribution bounds: extended abstract.
\newblock In Cynthia Dwork, editor, {\em STOC}, pages 599--608. ACM, 2008.

\bibitem{JainRS03}
Rahul Jain, Jaikumar Radhakrishnan, and Pranab Sen.
\newblock A direct sum theorem in communication complexity via message
  compression.
\newblock In Jos C.~M. Baeten, Jan~Karel Lenstra, Joachim Parrow, and
  Gerhard~J. Woeginger, editors, {\em ICALP}, volume 2719 of {\em Lecture Notes
  in Computer Science}, pages 300--315. Springer, 2003.

\bibitem{JainRS05}
Rahul Jain, Jaikumar Radhakrishnan, and Pranab Sen.
\newblock Prior entanglement, message compression and privacy in quantum
  communication.
\newblock In {\em IEEE Conference on Computational Complexity}, pages 285--296.
  IEEE Computer Society, 2005.

\bibitem{JayramW09}
T.~S. Jayram and David~P. Woodruff.
\newblock The data stream space complexity of cascaded norms.
\newblock In {\em FOCS}, pages 765--774. IEEE Computer Society, 2009.

\bibitem{KalyanasundaramS92}
Bala Kalyanasundaram and Georg Schnitger.
\newblock The probabilistic communication complexity of set intersection.
\newblock {\em SIAM J. Discrete Math.}, 5(4):545--557, 1992.

\bibitem{KarchmerW90}
Mauricio Karchmer and Avi Wigderson.
\newblock Monotone circuits for connectivity require super-logarithmic depth.
\newblock {\em SIAM J. Discrete Math.}, 3(2):255--265, 1990.

\bibitem{KleitmanKR71}
D.~L. Kleitman, M.~M. Krieger, and B.~L. Rotschild.
\newblock Configurations maximizing the number of pairs of {H}amming-adjacent
  lattice points.
\newblock {\em Studies in Appl. Math.}, 50:115--119, 1971.

\bibitem{KushilevitzN97}
Eyal Kushilevitz and Noam Nisan.
\newblock {\em Communication Complexity}.
\newblock Cambridge University Press, 1997.

\bibitem{Lindsey64}
John~H. Lindsey.
\newblock Assignment of numbers to vertices.
\newblock {\em The American Mathematical Monthly}, 71(5):508--516, 1964.

\bibitem{MagniezMN10}
Fr{\'e}d{\'e}ric Magniez, Claire Mathieu, and Ashwin Nayak.
\newblock Recognizing well-parenthesized expressions in the streaming model.
\newblock In Schulman \cite{DBLP:conf/stoc/2010}, pages 261--270.

\bibitem{Miltersen94}
Peter~Bro Miltersen.
\newblock Lower bounds for union-split-find related problems on random access
  machines, 1994.

\bibitem{MiltersenNSW98}
Peter~Bro Miltersen, Noam Nisan, Shmuel Safra, and Avi Wigderson.
\newblock On data structures and asymmetric communication complexity.
\newblock {\em J. Comput. Syst. Sci.}, 57(1):37--49, 1998.

\bibitem{MolinaroWY13}
Marco Molinaro, David Woodruff, and Grigory Yaroslavtsev.
\newblock Beating the direct sum theorem in communication complexity with
  implications for sketching.
\newblock In {\em Proceedings of 24th ACM-SIAM Symposium on Discrete
  Algorithms}, pages 1486--1502, 2013.

\bibitem{Newman91}
Ilan Newman.
\newblock Private vs. common random bits in communication complexity.
\newblock {\em Information Processing Letters}, 39(2):67 -- 71, 1991.

\bibitem{NisanW93}
Noam Nisan and Avi Wigderson.
\newblock Rounds in communication complexity revisited.
\newblock {\em SIAM J. Comput.}, 22(1):211--219, 1993.

\bibitem{ParnafesIRWA97}
Itzhak Parnafes, Ran Raz, and Avi Wigderson.
\newblock Direct product results and the gcd problem, in old and new
  communication models.
\newblock In {\em Proceedings of the twenty-ninth annual ACM symposium on
  Theory of computing}, STOC '97, pages 363--372, New York, NY, USA, 1997. ACM.

\bibitem{Patrascu09}
Mihai Pătraşcu.
\newblock Cc4: One-way communication and a puzzle.
\newblock
  \url{http://infoweekly.blogspot.com/2009/04/cc4-one-way-communication-and-pu%
zzle.html}.
\newblock Accessed: 31/03/2013.

\bibitem{Patrascu11}
Mihai Pătraşcu.
\newblock Unifying the landscape of cell-probe lower bounds.
\newblock {\em SIAM J. Comput.}, 40(3):827--847, 2011.

\bibitem{Razborov92}
Alexander~A. Razborov.
\newblock On the distributional complexity of disjointness.
\newblock {\em Theor. Comput. Sci.}, 106(2):385--390, 1992.

\bibitem{DBLP:conf/stoc/2010}
Leonard~J. Schulman, editor.
\newblock {\em Proceedings of the 42nd ACM Symposium on Theory of Computing,
  STOC 2010, Cambridge, Massachusetts, USA, 5-8 June 2010}. ACM, 2010.

\bibitem{Sen03}
Pranab Sen.
\newblock Lower bounds for predecessor searching in the cell probe model.
\newblock In {\em IEEE Conference on Computational Complexity}, pages 73--83.
  IEEE Computer Society, 2003.

\bibitem{Woodruff08}
David~P. Woodruff.
\newblock personal communication, 2008.

\end{thebibliography}
